\documentclass[twocolumn,aps,prd,amsmath,amsfonts,nofootinbib,showpacs]{revtex4-1}

\newcommand{\ud}{\mathrm{d}}

\newcommand{\half}{{\tfrac{1}{2}}}

\newcommand{\bs}{\begin{split}}
\newcommand{\es}{\end{split}}

\usepackage{verbatim}
\usepackage{float}
\usepackage{color}
\usepackage{breakurl}
\usepackage[hyperindex,breaklinks]{hyperref}
\hypersetup{
    pdfborder={0 0 0.5 0},
    linkbordercolor={0.6 1 0.6},
    citebordercolor={0.6 1 0.6},
    urlbordercolor={0.6 1 0.6},
    pdftitle={Traversable Wormholes and Classical Scalar Fields},
    pdfauthor={Luke M. Butcher}}
\usepackage{slashed}
\usepackage{graphicx}
\usepackage{subfigure}
\usepackage{psfrag}
\usepackage{mathtools}
\usepackage{amsthm}
\theoremstyle{plain}
\newtheorem*{Theorem}{Theorem}
\newtheorem*{Lemma}{Lemma}

\begin{document}

\title{Traversable Wormholes and Classical Scalar Fields}
\author{Luke M. Butcher}
\email[]{lmb@roe.ac.uk}
\affiliation{Institute for Astronomy, University of Edinburgh, Royal Observatory, Edinburgh EH9 3HJ, United Kingdom}
\date{13 March 2015}

\begin{abstract}
I prove that general relativity admits no asymptotically well-behaved static spherically symmetric traversable wormholes supported by classical scalar fields and nonexotic matter. The theorem holds for all values of the scalar field curvature coupling parameter $\xi$, even though fields with $\xi>0$ are capable of violating the average null energy condition. Good asymptotic behaviour (the effective Newton's constant being positive and finite at both ends of the wormhole) can only be achieved by introducing additional exotic matter, which must itself violate the null energy condition over some region of nonzero volume. These results are insensitive to the number of scalar fields, the form of their potentials, and the coupling between the fields and the additional matter.
\end{abstract}
\maketitle

\section{Introduction}\label{intro}
We do not know whether \emph{traversable wormholes} are allowed by the laws of nature. This gap in our knowledge represents a potential challenge to the role of \emph{causality} in physics, as there appears to be no fundamental obstacle that would prevent these `bridges' of curved spacetime being used as time machines \cite{Morris88}. Consequently, there is a strong desire to settle the matter one way or the other: to either provide a believable mechanism by which a traversable wormhole could form and be maintained, or to use known laws to formulate a rigorous argument that would forbid their existence altogether. 

The starting point for any discussion of this problem is the well-known fact that traversable wormholes are precluded by the \emph{average null energy condition} (ANEC)  \cite{Morris88x, Friedman93, Hochberg98}; in other words, traversable wormholes require matter with \emph{negative} energy, as averaged along a complete null geodesic. Although this demand is not easily met, it does not rule out the existence of traversable wormholes entirely, as there are a number of plausible mechanisms by which this exotic matter could be obtained. Theoretical investigations into negative energy have typically focused on \emph{quantum} phenomena, e.g.\ \cite{Casimir48, BirrellDavies, Anderson95, Khabibullin06, Butcher14}. However, as Barcel\'o and Visser have emphasised \cite{Barcelo02}, there are also credible \emph{classical} systems which violate the various energy conditions \cite{Bekenstein74, Bekenstein75, Deser84, Flanagan96}, and these have the advantage that there is no danger that some quantum inequality \cite{Ford95, Fewster08a, Fewster08b} might constrain the negative energy and render it useless for wormhole maintenance.

Of particular interest is the classical scalar field, the dynamics of which are determined by the action
\begin{align}\label{action}
S_\varphi = -\frac{1}{2}\int\! \ud^4 x \sqrt{-g}\left(\left(\nabla \varphi\right)^2 +2V(\varphi)+ \xi R \varphi^2\right),
\end{align}
where $V$ is the potential and $\xi$ the curvature coupling parameter.\footnote{We set $c=1$ and adopt the sign conventions of Wald \cite{Wald}: $\eta_{\mu\nu}\equiv \mathrm{diag}(-1,1,1,1)$, $[\nabla_\mu,\nabla_\nu]v^\alpha\equiv R^{\alpha}_{\phantom{\alpha}\beta\mu\nu}v^\beta$, and $R_{\mu\nu}\equiv R^{\alpha}_{\phantom{\alpha}\mu\alpha\nu}$.} Given that these uncontroversial fields play a central role in the standard model of particle physics and cosmology, and that nonminimal coupling (particularly $\xi=1/6$) is physically well-motivated,\footnote{Despite the superficial simplicity of `minimal' coupling $\xi=0$, this choice is unjustified by any symmetry or special property, and is unstable to quantum corrections \cite{Ford82}. In contrast, there are good reasons to believe that $\xi=1/6$ is actually the most natural possibility. Firstly, this is a \emph{conformal} coupling, meaning that the dynamics of the massless field ($V=0$) are invariant under conformal transformations $g_{\mu\nu}\to \Omega^2 g_{\mu\nu}$; this same invariance is displayed by the electromagnetic field, and thus serves as a sensible definition of massless behaviour in curved space. Secondly, it is only for $\xi=1/6$ that the field has an energy-momentum tensor with good renormalisation properties, corresponding to the `new improved' energy-momentum of Callan, Coleman and Jackiw \cite{Callan70}.} it is perhaps surprising to learn that the scalar field energy-momentum tensor
 \begin{align}\nonumber
T^\varphi_{\mu\nu}&\equiv \frac{-2}{\sqrt{-g}}\frac{\delta S_\varphi}{\delta g^{\mu\nu}}\\\nonumber
&= \nabla_\mu \varphi \nabla_\nu \varphi - \half g_{\mu\nu} (\nabla \varphi)^2 - g_{\mu\nu}V(\varphi)\\\label{Tphi}&\quad {} + \xi \left(G_{\mu\nu}\varphi^2 - \nabla_\mu\nabla_\nu (\varphi^2) + g_{\mu\nu} \nabla^2 (\varphi^2)\right),
\end{align}
can violate all the pointwise energy-conditions (null, weak, dominant, strong) under various circumstances, and (provided $R_{\mu\nu}\ne0$ \cite{Fewster06} and $\xi >0$) can even violate the ANEC \cite{Barcelo99, Barcelo00}.

Motivated by this last result, Barcel\'o and Visser were able to construct solutions to the Einstein field equations representing static spherically symmetric traversable wormholes supported by a massless classical scalar field, firstly for conformal coupling ($\xi=1/6$) \cite{Barcelo99}, and then more generally ($\xi>0$) \cite{Barcelo00}. Although this constituted a promising step towards providing a convincing physical mechanism by which a traversable wormhole could be maintained, the solutions presented in these papers have a rather strange and conspicuous feature: at one end of the wormhole, the coupling between matter and gravity is reversed. To understand how this peculiar behaviour arises, suppose we have a scalar field and some additional matter with energy-momentum $T^\mathrm{add}_{\mu\nu}$. As usual, the Einstein field equations will be $G_{\mu\nu}= \kappa(T_{\mu\nu}^\varphi+ T^\mathrm{add}_{\mu\nu})$, where $\kappa=8\pi G$ is the gravitational coupling (Newton's constant). Notice, however, that the scalar field energy-momentum tensor (\ref{Tphi}) contains a term $\xi G_{\mu\nu}\varphi^2$; we must group this term with the Einstein tensor on the left of the field equations, otherwise the actual relationship between matter and curvature will be obscured.  The resulting field equations are $G_{\mu\nu}=\kappa_\mathrm{eff}(T^{\varphi\prime}_{\mu\nu}+ T^\mathrm{add}_{\mu\nu})$, where $T^{\varphi\prime}_{\mu\nu} \equiv T^{\varphi}_{\mu\nu} -\xi G_{\mu\nu}\varphi^2$ is the part of the scalar field energy-momentum that does not depend on $G_{\mu\nu}$, and 
\begin{align}\label{keff}
\kappa_\text{eff}\equiv \frac{1}{\kappa^{-1}-\xi \varphi^2}
\end{align}
is the \emph{effective} gravitational coupling. In each wormhole solution described by Barcel\'o and Visser, the scalar field increases monotonically from one end to the other, reaching an asymptotic limit with $\xi \varphi^2>\kappa^{-1}$. This end of the wormhole will therefore experience a negative effective gravitational coupling constant $\kappa_\text{eff}<0$; under these conditions, both $T^{\varphi\prime}_{\mu\nu}$ and $T^\mathrm{add}_{\mu\nu}$ induce spacetime curvature with the \emph{opposite} sign to usual, giving rise to a type of `anti-gravity' effect. Even if this situation does not induce any paradoxical or mathematically pathological behaviour, it is certainly clear that these wormholes do not describe a connection between two parts of our universe, nor even a connection to another universe remotely similar to our own.\footnote{The regions of negative gravitational coupling also give rise to instabilities \cite{Bronnikov01, Bronnikov06}, although it may be possible to neutralise this effect using nonexotic matter with suitable dynamical properties. We need not concern ourselves with questions of stability in this article: the unperturbed behaviour of the wormhole will prove sufficiently problematic by itself.}

Barcel\'o and Visser claim that this problem can be overcome by introducing shells of nonexotic matter: ``if we add to our system a quantity of normal matter\dots\,in two thin spherical shells, we can join smoothly\dots\,two outer asymptotic geometries, both with positive effective Newton's constants'' \cite{Barcelo99}. Unfortunately, this assertion is false. The next section will prove that ordinary matter (in thin shells or otherwise) cannot allow a classical scalar field to support a static spherically symmetric traversable wormhole such that both asymptotic regions have positive finite effective gravitational coupling. As we will see, the additional matter required will always violate the null energy condition over a region of nonzero volume. In other words, the energy-condition violations of the scalar field have been in vain, and it is always necessary to invoke yet more exotic matter to construct a static spherically symmetric traversable wormhole with good asymptotic properties. 

\section{Argument}\label{argument}
\begin{Theorem}
Let us express the metric of a static spherically symmetric traversable wormhole as
\begin{align}\label{metric}
\ud s^2 = - A(r)\ud t^2 + \ud r^2 + B(r)\left(\ud \theta^2 + \sin^2\!\theta\,\ud \phi^2 \right),
\end{align}
with
\begin{align}\label{AB>0}
A(r)\ge A_\mathrm{min}>0,\quad B(r)\ge B_\mathrm{min}>0\quad\ \forall r \in \mathbb{R},
\end{align}
and define on this spacetime a (static spherically symmetric) classical scalar field $\varphi$ and some additional matter with energy-momentum tensor $T^\mathrm{add}_{\mu\nu}$. Assume also that $B$ is not a constant, and that in both asymptotic regions the effective gravitational coupling (\ref{keff}) is positive and finite:
\begin{align}\label{keffinf}
0<\lim_{r\to\pm \infty}\kappa_\text{eff}<\infty.
\end{align}
Then in order to satisfy the Einstein field equations,
\begin{align}\label{EFE}
G_{\mu\nu}= \kappa \left(T^{\varphi}_{\mu\nu}+ T^\mathrm{add}_{\mu\nu}\right),
\end{align}
with $T^{\varphi}_{\mu\nu}$ given by (\ref{Tphi}), the additional matter must violate the null energy condition:
\begin{align}\label{NECviolate}
\exists k^\mu \quad \text{s.t.}\quad k^\mu k_\mu=0\quad\text{and}\quad T^\mathrm{add}_{\mu\nu}k^{\mu}k^{\nu} <0,
\end{align}
throughout some region of nonzero measure.
\end{Theorem}

Before proving this result, it is worth making a few remarks:
\begin{enumerate}
\item Coordinates can always be found to express a static spherically symmetric traversable wormhole as we have in (\ref{metric}) and (\ref{AB>0}); hence this particular form of the metric does not place any additional constraint on the applicability of the theorem.

\item The existence of $A_\mathrm{min}$ and $B_\mathrm{min}$ allows us to avoid pathological spacetimes that `pinch off' ($A\to 0^+$ or $B\to 0^+$) at either end. Similarly, that $B$ is not a constant rids us of a trivial case: an infinitely long throat that never `flares out'. Neither of these situations represent genuine wormholes. It goes without saying that both conditions are automatically satisfied if the wormhole is asymptotically flat, wherein $A\sim O(1)$ and $B\sim O(r^2)$ as $r\to \pm \infty$.

\item If the wormhole connects two different locations within the \emph{same} universe (so that there is really only one asymptotic limit) then we should of course insist that $\lim_{r\to\infty}\kappa_\text{eff}=\lim_{r\to-\infty}\kappa_\text{eff}$. We will not need this stipulation to prove the theorem. 

\item Nothing is assumed about the value of the scalar field's curvature-coupling parameter $\xi$, nor the form of its potential $V(\varphi)$. Moreover, there is no requirement that $\varphi$ obey a particular field equation; hence the proof will hold regardless of whether $\varphi$ is coupled in some way to the additional matter.

\item One does not need to include an explicit cosmological term $g_{\mu\nu}\Lambda$ in the Einstein field equations (\ref{EFE}); this supposed generalisation is already covered by the freedom to add a constant to $V(\varphi)$.

\item There is an implicit assumption that $A$, $B$ and $\varphi$ are everywhere twice differentiable, as this is required for the existence of $G_{\mu\nu}$ and $T^{\varphi}_{\mu\nu}$. Strictly speaking, then, the theorem does not apply to wormholes constructed using a cut-and-paste approach (see \cite{VisserBook} and references therein) where \emph{infinitesimal} shells of matter (delta-functions) allow for discontinuities in $B'$, $A'$ and $\varphi'$. However, it is important to realise that delta-function shells are only a mathematical idealisation: a \emph{real} shell of matter must have some \emph{nonzero} thickness, and in this case $A''$, $B''$, and $\varphi''$ will exist as required. The theorem therefore applies to any \emph{actual} configuration of matter, and in particular, to any physical realisation of a cut-and-paste wormhole.

\item The `measure' of a region $\mathcal{D}$ is $\int_\mathcal{D}\sqrt{-g} \ud^4 x$. As the system under consideration is static and spherically symmetric, we can restrict our interest to regions of the form $\mathcal{D}= \{(t,r,\theta,\phi): r\in D\}$ for some $D\subseteq \mathbb{R}$. Under this restriction, and recalling (\ref{AB>0}), it is easy to see that $\mathcal{D}$ having nonzero measure is equivalent to $\int_D \ud r \ne 0$. Thus the theorem indicates that $T^\mathrm{add}_{\mu\nu}$ does not simply violate the null energy condition on a set of isolated radii (hypersurfaces of constant $r$) but over at least one continuous range $r\in (r_1,r_2)$.
\end{enumerate}
With the remarks concluded, we now prove the proposition.

\begin{proof}

Given a metric (\ref{metric}) and scalar field $\varphi$, the Einstein field equations (\ref{EFE}) determine the energy-momentum that the addition matter must supply:
\begin{align}\nonumber
T^\mathrm{add}_{\mu\nu} &=\kappa^{-1}G_{\mu\nu}  - T^{\varphi}_{\mu\nu}\\\nonumber
&=(\kappa^{-1}- \xi \varphi^2)G_{\mu\nu} - \nabla_\mu \varphi \nabla_\nu \varphi +\half g_{\mu\nu} (\nabla \varphi)^2\\
&\quad {}+ g_{\mu\nu}V(\varphi) + \xi \left(\nabla_\mu\nabla_\nu (\varphi^2) - g_{\mu\nu} \nabla^2 (\varphi^2)\right).
\end{align}
Each null vector $k^\mu$ then defines an energy-density for the additional matter,
\begin{align}\nonumber
\rho^\mathrm{add}&\equiv T^\mathrm{add}_{\mu\nu}k^\mu k^\nu\\\nonumber
&=(\kappa^{-1}- \xi \varphi^2)R_{\mu\nu}k^\mu k^\nu + 2\xi \varphi k^\mu k^\nu\nabla_\mu\nabla_\nu \varphi,\\
&\quad {} +(2\xi -1)(k^\mu \nabla_\mu \varphi)^2.
\end{align}
Recalling that the scalar field is static and spherically symmetric $\varphi=\varphi(r)$, and noting that $\Gamma^r{}_{\mu\nu}= - \partial_r g_{\mu\nu}/2$ for the particular metric (\ref{metric}), the above formula simplifies to
\begin{align}\nonumber
\rho^\mathrm{add}&=(\kappa^{-1}- \xi \varphi^2)R_{\mu\nu}k^\mu k^\nu  +\xi \varphi \varphi' k^\mu k^\nu \partial_r g_{\mu\nu}
\\\label{rho}
&\quad {} + (k^r)^2 \left(2 \xi \varphi \varphi'' +(2\xi -1) (\varphi')^2 \right),
\end{align}
where primes denote differentiation with respect to $r$. To prove that the additional matter violates the null energy condition (\ref{NECviolate}) we will show that there is always some choice of $k^\mu$ for which $\rho^\mathrm{add}<0$ throughout some region of nonzero measure. 

At this point it is useful to distinguish between two possibilities for the scalar field:
\begin{align}\label{case1}
 &\text{either}& \kappa^{-1} -\xi \varphi^2&\ge0 \quad\text{everywhere}, \\ \label{case2}
& \text{or} &\kappa^{-1} - \xi \varphi^2&<0\quad \text{somewhere};
\end{align}
we will treat these two cases separately, using a different null vector $k^\mu$ for each. Note that the asymptotic condition (\ref{keffinf}) is equivalent to
\begin{align}\label{phiinf}
\kappa^{-1} - \xi\varphi^2(\pm\infty)>0,
\end{align}
where $\varphi(\pm\infty)\equiv \lim_{r\to \pm \infty}\varphi(r)$ are the asymptotic values of the scalar field. Thus, in the second case (\ref{case2}), the continuity of $\varphi$ guarantees the existence of (at least) two radii $r_-<r_+$ such that
\begin{align}\label{rpm}\begin{split}
\kappa^{-1} - \xi \varphi^2(r_{\pm})&=0, \\
\text{and} \quad \kappa^{-1} - \xi \varphi^2(r)&<0\quad \forall r \in (r_-,r_+).
\end{split}
\end{align}
We will make use of these radii near the end of the proof.

For now, let us focus on the first case (\ref{case1}) and consider a longitudinal null vector, pointing directly along the throat:
\begin{align}\label{klong}
k^\mu_\mathrm{long}\equiv (A^{-1/2},1,0,0).
\end{align}
The energy-density (\ref{rho}) is then
\begin{align}\nonumber
\rho^\mathrm{add}_\mathrm{long}&=(\kappa^{-1}- \xi \varphi^2)\left(\frac{R_{tt}}{A} +R_{rr}\right) -\xi \varphi \varphi^\prime \frac{A^\prime}{A}
\\&\quad {}  + 2\xi \varphi \varphi'' +(2\xi -1) (\varphi')^2.
\end{align}
As the nonzero components of the Ricci tensor are given by
\begin{align}\label{ricci}\begin{split}
\frac{R_{tt}}{A}&= \frac{A''}{2A} - \frac{1}{4}\left(\frac{A'}{A}\right)^2 + \frac{A'B'}{2AB},\\
R_{rr}&=- \frac{A''}{2A} +\frac{1}{4}\left(\frac{A'}{A}\right)^2+ \frac{1}{2}\left(\frac{B'}{B}\right)^2 - \frac{B''}{B},\\
\frac{R_{\theta\theta}}{B}&=\frac{R_{\phi\phi}}{B \sin^2 \theta}= - \frac{A'B'}{4AB}  - \frac{B''}{2B} + \frac{1}{B},
\end{split}
\end{align}
we have
\begin{align}\nonumber
\rho^\mathrm{add}_\mathrm{long}&=(\kappa^{-1}- \xi \varphi^2)\left(\frac{A'B'}{2AB}+ \frac{1}{2}\left(\frac{B'}{B}\right)^2 - \frac{B''}{B} \right)\\
&\quad {} -\xi \varphi \varphi^\prime \frac{A^\prime}{A} + 2\xi \varphi \varphi'' +(2\xi -1) (\varphi')^2,
\end{align}
and it is easy to verify that this can be rewritten as
\begin{align}\nonumber
\rho^\mathrm{add}_\mathrm{long}&=-\frac{3}{2}(\kappa^{-1}- \xi \varphi^2)\left(\frac{B'}{B}\right)^2 -  (\varphi')^2\\\label{I-}
&\quad {} + \sqrt{A} B\frac{\ud}{\ud r}\left[\frac{-1}{\sqrt{A}B^2} \frac{\ud}{\ud r} \left[(\kappa^{-1}- \xi \varphi^2) B\right]\right].
\end{align}
Now consider the following integral:
\begin{align}\nonumber
I& \equiv \int^{R_+}_{R_-}\rho^\mathrm{add}_\mathrm{long}\frac{\ud r}{\sqrt{A}B} \\\nonumber
&= -\int^{R_+}_{R_-}\left( \frac{3}{2}(\kappa^{-1}- \xi \varphi^2)\left(\frac{B'}{B}\right)^2 + (\varphi')^2 \right)\frac{\ud r}{\sqrt{A}B}\\\label{I}
&\quad {} + \left[\frac{-1}{\sqrt{A}B^2} \frac{\ud}{\ud r} \left[(\kappa^{-1}- \xi \varphi^2) B\right]\right]^{R_+}_{R_-}.
\end{align}
As a consequence of (\ref{case1}), the function 
\begin{align}
f(r)\equiv (\kappa^{-1}-\xi \varphi^2 ) B\ge0
\end{align}
is bounded from below, and it is therefore possible to send $R_{\pm}\rightarrow \pm \infty$ in such a way that the boundary term $[(-1/\sqrt{A}B^2)f'(r)]^{R_+}_{R_-}$ either tends to zero, or is negative.\footnote{To prove this result in full generality requires a little analysis, but the main idea can be explained without complication when the limits $f'(\pm\infty)\equiv \lim_{r\rightarrow \pm \infty} f'(r)$ exist. It is easy to see that $f'(\infty)<0$ contradicts the lower bound for $f$ (because $\lim_{r\rightarrow \infty} f(r) = \lim_{r\rightarrow \infty}\int^r f'(x)\ud x=-\infty$) and that $f'(-\infty)>0$ leads to the same contradiction. Thus $[f'(r)]^{+\infty}_{-\infty}= f'(\infty)-f'(-\infty)\ge0$. We then note that (\ref{AB>0})  ensures that $-1/\sqrt{A}B^2$ is negative and bounded from below. Consequently, the boundary term $[(-1/\sqrt{A}B^2)f'(r)]^{+\infty}_{-\infty}\le 0$, as required.

If the limits $f'(\pm\infty)$ do not exist, then one should instead use unbounded sequences $R_{+1}<R_{+2}<\ldots$  and $R_{-1}>R_{-2}>\ldots$ such that $f'(R_{+n})$ is either positive for all $n$, or tends to zero as $n\to \infty$, and $f'(R_{-n})$ is either negative for all $n$, or tends to zero as $n\to \infty$. These sequences, and aforementioned properties of $-1/\sqrt{A}B^2$, allow one to send $R_{\pm} \to \pm\infty$ with a boundary term that is either negative or vanishingly small. For the sake of completeness, the existence of these sequences is proved in the appendix. \label{boundary}}
As a result,
\begin{align}
I \le -\int^{\infty}_{-\infty}\!\!\left( \frac{3}{2}(\kappa^{-1}- \xi \varphi^2)\left(\frac{B'}{B}\right)^2 + (\varphi')^2 \right)\!\!\frac{\ud r}{\sqrt{A}B}.
\end{align}
Referring to (\ref{case1}) once again, we see that neither term in the integrand can ever be negative, so we must have $I\le0$. In fact, $I$ can only be zero if (i) $\kappa^{-1} - \xi \varphi^2 =0$ wherever $B'\ne0$ (which must be somewhere, as $B$ is not a constant), and (ii) $\varphi'=0$ everywhere; this requires that $\kappa^{-1} - \xi \varphi^2 =0$ everywhere, which contradicts the asymptotic condition (\ref{phiinf}). Hence $I<0$, and we conclude that there must be some region, of nonzero measure, where $\rho^\mathrm{add}_\mathrm{long}<0$.

We now turn to the second case (\ref{case2}), and consider a null vector
\begin{align}\label{kcirc}
k^\mu_\mathrm{circ}\equiv(A^{-1/2},0,0,B^{-1/2}(\sin{\theta})^{-1}),
\end{align}
which circulates around the throat, running perpendicular to the longitudinal direction. The circulating null vector defines an energy-density (\ref{rho}) as follows:
\begin{align}\nonumber
\rho_\mathrm{circ}^\mathrm{add} &= (\kappa^{-1} - \xi \varphi^2)\left(\frac{R_{tt}}{A} + \frac{R_{\phi\phi}}{B \sin^2 \theta}\right) \\\nonumber
&\quad {}+ \xi \varphi \varphi' \left( \frac{B'}{B} - \frac{A'}{A}\right)\\\nonumber
&=(\kappa^{-1} - \xi \varphi^2)\Bigg(\frac{A''}{2A}- \frac{1}{4} \left(\frac{A'}{A}\right)^2 + \frac{A'B'}{4AB} \\
&\quad {} - \frac{B''}{2B} + \frac{1}{B}\Bigg) + \xi \varphi \varphi' \left( \frac{B'}{B} - \frac{A'}{A}\right),
\end{align}
where (\ref{ricci}) was used in the last line. We then notice this can be rewritten as
\begin{align}\nonumber
\rho_\mathrm{circ}^\mathrm{add} &=\frac{1}{2\sqrt{A} B}\frac{\ud}{\ud r}\left[(\kappa^{-1} -\xi \varphi^2)\left(\frac{A' B}{\sqrt{A}} - \sqrt{A}B'\right)\right] \\ \label{J-}
&\quad {} + \frac{\kappa^{-1} - \xi \varphi^2}{B},
\end{align}
which inspires us to construct the integral
\begin{align}\nonumber
J &\equiv \int^{r_+}_{r_-} \rho^\mathrm{add}_\mathrm{circ}\sqrt{A} B \ud r\\\nonumber
& = \int^{r_+}_{r_-} (\kappa^{-1} -\xi \varphi^2) \sqrt{A} \ud r \\\label{J}
&\quad {}+ \frac{1}{2}\left[(\kappa^{-1} -\xi \varphi^2)\left(\frac{A' B}{\sqrt{A}} - \sqrt{A}B'\right)\right]^{r_+}_{r_-}.
\end{align}
As we have chosen the radii $r_{\pm}$ for the limits of integration, we see from (\ref{rpm}) that the integrand is strictly negative, and the boundary term is zero. Hence $J<0$, and there must therefore be a region $D\subseteq (r_{-},r_{+})$, with nonzero measure, throughout which $\rho^\mathrm{add}_\mathrm{circ}<0$.
\end{proof}

\section{Discussion}
Even when it violates the average null energy condition, the classical scalar field (\ref{action}) cannot support a static spherically symmetric traversable wormhole (\ref{metric}) with sensible asymptotic behaviour (\ref{keffinf}). We have shown that, in addition to the scalar field, these wormholes require matter that must itself violate the null energy condition (over a region of nonzero volume) defeating the whole purpose of using a scalar field in the first place.

It is worth noting that although the theorem refers to a \emph{single} scalar field, it can easily be generalised to wormholes populated by an arbitrary number of fields $\{\varphi_1, \varphi_2\ldots \}$, each with a distinct potential $\{V_1, V_2\ldots \}$ and curvature coupling parameter $\{\xi_1, \xi_2\ldots\}$. The proof proceeds exactly as above, the only difference being that the scalar field energy-momentum tensor is replaced by a sum of contributions from each field:
\begin{align}
T^\varphi_{\mu\nu}\to \sum_n T^{\varphi_n}_{\mu\nu}.
\end{align}
This change can be realised in each line of the argument simply by making the replacements $\varphi\to \varphi_n$, $V\to V_n$, $\xi\to\xi_n$, and taking a sum $\sum_n$ over every term that depends on $n$. It is easy to check that these alterations do not interfere with the logic of the proof, so long as they are applied consistently throughout. Consequently, there can be no way to use additional classical scalar fields to provide the additional exotic matter $T^\mathrm{add}_{\mu\nu}$ needed by the wormhole.

Speaking less formally, the theorem also has implications for wormholes that are neither static nor spherically symmetric, casting substantial doubt on the notion that one can use a classical scalar field to support \emph{any} asymptotically well-behaved traversable wormhole. This implication arises because deviations from perfect symmetry typically make it \emph{even harder} to maintain a traversable wormhole: asymmetric modes often give rise to instabilities, and tend to increase the energy of the scalar field, deepening the need for exotic matter. Although a serious attempt to generalise the theorem to generic traversable wormholes is beyond the scope of this article, we suggest that progress might be made by averaging the Einstein field equations over $\{t, \theta, \phi\}$; this will reduce a more generic wormhole (with $A(t,r,\theta,\phi)$, $B(t,r,\theta,\phi)$ and $\varphi(t,r,\theta,\phi)$) to a 1-dimensional system, which should then be susceptible to the techniques used here. It may also be possible to adapt to new terms in the metric (such as those proportional $\ud t \ud \phi$ or $\ud r \ud \phi$, corresponding to `spinning' or `twisted' wormholes) by making more elaborate specifications for the null vector $k^\mu$, beyond simply the longitudinal (\ref{klong}) and circulating (\ref{kcirc}) definitions used above.

The broad view one draws from this work, then, is that classical scalar fields are poor candidates for maintaining traversable wormholes: they can certainly violate the average null energy condition, but it appears they cannot do so \emph{usefully}. Given that this is a failure of a \emph{classical} field, should we infer that only \emph{quantum} systems are capable of providing the sort of exotic matter we need? While this is a tempting conclusion, it is important to recognise that quantum fields have demonstrated similar shortcomings so far: the quantum scalar field violates the ANEC without much effort (and indeed, without reversing the sign of $\kappa_\mathrm{eff}$ anywhere) but it has proven extremely difficult to use this violation for the desired purpose of actually stabilising a wormhole. As a recent example of this, wormholes with very long throats were found to induce in the conformally-coupled quantum scalar field a vacuum state which violates the ANEC along \emph{almost every} null geodesic, but even this could not prevent the wormhole's collapse \cite{Butcher14}. This recurrent frustration suggests that the ANEC may not be the fundamental obstacle to maintaining traversable wormholes, and hints at the existence of some other principle or energy condition, as yet unformulated, which should be used instead. Such a principle would (like the ANEC) demonstrably preclude the existence of traversable wormholes, but (unlike the ANEC) would actually be \emph{obeyed} by both quantum and classical scalar fields, assuming sensible asymptotic behaviour. This principle would explain the invisible obstruction we have encountered in our attempts to use scalar field exotic matter to support traversable wormholes.

Given the critical role played by the energy-density integrals $I$ and $J$ in the proof, it is worth considering whether these constructions might be useful for other arguments, and indeed, whether they have some relevance to the fundamental principle alluded to above. We note that $I$ and $J$ superficially resemble the ANEC integral, in that they are (unnormalised) averages of an energy-density $T_{\mu\nu}k^\mu k^\nu$ with $k^\mu k_\mu =0$; however, in departure from the ANEC, $I$ and $J$ have no obvious connection with \emph{geodesics}. In fact, the integrand of $J$ is strongly suggestive of a \emph{volume} integral $\int\ldots 4 \pi B \ud r$, with a weighting $\sqrt{A}$ that accounts for the redshift (or blueshift) of energy with respect to infinity. Therefore, in contrast to the volume integrals commonly used to quantify the exotic matter required by wormholes \cite{Visser03,Kar04,Lobo05} (which do not include a redshift factor) $J$ more closely resembles a contribution to the \emph{total mass} of the wormhole (i.e. the Komar, Bondi, or ADM masses) with the unusual feature being that it integrates an energy-density defined by null vectors $k^\mu_\mathrm{circ}$ that loop around the throat at fixed values of $r$. On the other hand, $I$ is far more closely related to the ANEC integral: firstly, the energy-density is obtained using the longitudinal null vector $k^\mu_\mathrm{long}$, the geodesics of which cover the full range of integration $-\infty<r<\infty$; secondly, the measure $\ud r/\sqrt{A} B$ includes the factor $1/\sqrt{A}$ that naturally arises in the ANEC by virtue of the geodesic's affine parameterisation. In fact, the only difference between $I$ and the ANEC integral is the unusual factor of $1/B$; it is not clear whether this new feature has an intuitive interpretation.

Regardless of the physical meaning of $I$ and $J$, it is clear that the argument presented in this paper was crucially dependent on our freedom to define energy integrals using null vectors and weightings that did not correspond to affinely parameterised geodesics. With this in mind, it seems likely that there exists a more nuanced measure of exotic matter, whether a volume integral or otherwise, that could supplant the ANEC and allow us to properly discern the dividing-line between the negative energy that nature does abide, and the minimal requirements of traversable wormholes.

\begin{acknowledgments}
The author is supported by a research fellowship from the Royal Commission for the Exhibition of 1851; he thanks John Peacock for his helpful remarks, and the anonymous reviewers for their useful comments.
\end{acknowledgments}

\appendix*

\section{Existence of Particular Sequences}\label{analysis}
What follows is a demonstration of the existence of the sequences described at the end of footnote \ref{boundary}. The $\{R_{+n}\}$ can be identified with whichever sequence (either $\{x_n\}$ or $\{y_n\}$, detailed below) can be continued indefinitely. It is trivial to adapt the argument for $\{R_{-n}\}$.

\begin{Lemma}
For any differentiable function $f: \mathbb{R} \to \mathbb{R}$ that is bounded from below, there is an unbounded sequence $x_1<x_2<\ldots$ with 
\begin{align}
 \lim_{n\to\infty}f'(x_n) = 0,
\end{align}
or an unbounded sequence $y_1<y_2<\ldots$  with 
\begin{align}
 f'(y_n) \ge0\quad \forall n \in \mathbb{N}.
\end{align}
\end{Lemma}
\begin{proof}
One can attempt to construct these sequences as follows. For the first series, define $x_0=0$, and set each $x_{n+1}$ to be some point $x_{n+1}> x_n +1$ where $|f'(x_{n+1})|\le 1/|x_n +1|$. For the second series, define $y_0=0$, and set each $y_{n+1}$ to be some point $y_{n+1}> y_n +1$ where $f'(y_{n+1})\ge 0$. The only way that both these methods can fail is if there is some $n$ and some $m$ for which no such $x_{n+1}$ and $y_{m+1}$ exist:
\begin{align}\nonumber
|f'(x)|&> \frac{1}{|x_n +1|}& \forall x &> x_n+1,\\
\text{and}\qquad f'(x)&<0& \forall x &> y_m+1.
\end{align}
This means
\begin{align}
f'(x)&<- K<0,& \forall x &> R,
\end{align}
where $K\equiv 1/|x_n +1|$ and $R\equiv\max\{x_n+1,y_n+1\}$. Consequently, for all $x>R$ we have
\begin{align}\nonumber
f(x)&= f(R)+ \int^x_R f'(r)\ud r\\
&< f(R) -K(x-R),
\end{align}
and hence
\begin{align}
\lim_{x\to \infty} f(x)=-\infty,
\end{align}
which clearly contradicts the assumption that the function is bounded from below. We conclude that at least one of the above methods must continue indefinitely.
\end{proof}

\bibliography{ScalarW}
\end{document}